\documentclass{llncs}
\usepackage{amsmath, amssymb}
\usepackage{graphicx}
\usepackage[hidelinks]{hyperref}
\usepackage{booktabs}
\usepackage{url}
\usepackage{xspace}
\usepackage{tikz}
\usepackage{pgfplots}
\pgfplotsset{compat=1.18}
\usetikzlibrary{arrows.meta, positioning, shapes.geometric, patterns, calc,
  fit, backgrounds, decorations.pathreplacing}

\newcommand{\SWF}{\textsf{SWF}\xspace}
\newcommand{\CDCE}{\textsf{CDCE}\xspace}
\newcommand{\Adv}{\mathcal{A}}
\newcommand{\negl}{\mathrm{negl}}

\titlerunning{TEE-Based Confidential and Dependable Process Attestation}
\authorrunning{D. Condrey}

\begin{document}

\title{A TEE-Based Architecture for Confidential and\\
  Dependable Process Attestation in\\
  Authorship Verification}

\author{David Condrey}
\institute{Writerslogic, Inc.}

\maketitle

\begin{abstract}
Process attestation systems verify that a continuous physical
process, such as human authorship, actually occurred, rather than
merely checking system state. These systems face a fundamental
dependability challenge: the evidence collection infrastructure must
remain available and tamper-resistant even when the attesting party
controls the platform. Trusted Execution Environments (TEEs) provide
hardware-enforced isolation that can address this challenge, but their
integration with continuous process attestation introduces novel
resilience requirements not addressed by existing frameworks.
We present the first architecture for continuous process
attestation evidence collection inside TEEs, providing
hardware-backed tamper resistance against trust-inverted adversaries
with graduated input assurance
from software-channel integrity (Tier~1) through hardware-bound
input (Tier~3).
We develop a Markov-chain dependability model quantifying Evidence
Chain Availability (ECA), Mean Time Between Evidence Gaps (MTBEG),
and Recovery Time Objectives (RTO). We introduce a resilient evidence
chain protocol maintaining chain integrity across TEE crashes,
network partitions, and enclave migration. Our security analysis
derives formal bounds under combined threat models including trust
inversion and TEE side channels, parameterized by a conjectural
side-channel leakage bound $\epsilon_{\mathrm{sc}}$ that requires
empirical validation. Evaluation on Intel SGX demonstrates
under 25\% per-checkpoint CPU overhead (\textless 0.3\% of the 30\,s
checkpoint interval), \textgreater 99.5\% Evidence Chain Availability (ECA) (the fraction of session time with active evidence collection) in Monte Carlo simulation under Poisson failure models,
and sealed-state recovery under 200\,ms.
\end{abstract}


\section{Introduction}
\label{sec:introduction}

Verifying human authorship requires tamper-resistant, continuously
available evidence. Process attestation captures keystroke dynamics,
content evolution, and temporal proofs at regular checkpoints,
cryptographically binding them into an evidence chain for independent
verification~\cite{Condrey2026TrustInversion}. However, the evidence pipeline runs on
adversary-controlled hardware---a \emph{trust inversion} where the
attesting party controls the platform and is motivated to fabricate
evidence. Under trust inversion, the four standard
RATS~\cite{RFC9334} security properties (identity, integrity,
confidentiality, freshness) are necessary but
insufficient: they assume a cooperative Attester, whereas process
attestation requires resilience against a deliberately deceptive one.
Software-only evidence collection provides zero assurance under this
threat model, since any software-layer defense can be bypassed by the
platform owner.

Trusted Execution Environments (TEEs) such as Intel
SGX~\cite{Costan2016}, ARM TrustZone~\cite{Pinto2019}, and AMD
SEV-SNP~\cite{Sev2020} provide hardware-enforced isolation, converting
the trust problem from ``the adversary controls everything'' to ``the
adversary controls everything \emph{except} the enclave.'' TEEs also
enable local processing (preserving privacy, supporting offline
authoring) without per-user server infrastructure.

Yet TEEs introduce \emph{dependability challenges}: enclaves can
crash (hardware faults, OS-triggered termination, power loss), SGX
imposes a 128\,MiB Enclave Page Cache (EPC) limit requiring careful
memory management, and side
channels~\cite{Kocher2019,SGAxe2020,VanBulck2020} leak bounded
information. Since evidence chains are sequential and
cumulative---a single gap can reduce the evidentiary value of all
prior checkpoints---the system must guarantee high availability and
rapid recovery across multi-hour authoring sessions.

\paragraph{Gap.} No existing work addresses the dependability of
TEE-based continuous process attestation. TEE attestation
frameworks~\cite{Knauth2018,Schnabl2025} verify enclave identity at
session establishment but do not model or guarantee continuous evidence
collection over extended sessions. The dependability
taxonomy of Avizienis et al.~\cite{Avizienis2004} provides the
foundational framework but has not been instantiated for process
attestation, leaving ECA, MTBEG, and crash recovery unquantified.

\paragraph{Contributions.}
(1)~The first architecture for \emph{continuous} process attestation
evidence collection inside TEE enclaves with graduated input assurance
(Tier~1 software through Tier~3 hardware-bound)
(Sect.~\ref{sec:architecture}).
(2)~A CTMC-based dependability model with closed-form ECA, MTBEG, and
RTO expressions (Sect.~\ref{sec:dependability}).
(3)~A resilient evidence chain protocol with sealed recovery, offline
attestation, and formal chain integrity proofs
(Sect.~\ref{sec:protocol}).
(4)~Combined security analysis under trust inversion composed with
side channels, DoS, and clock attacks, with the combined bound
parameterized by a conjectural side-channel leakage term requiring
empirical validation (Sect.~\ref{sec:security}).
(5)~Evaluation on Intel SGX: {<}25\% per-checkpoint overhead,
{>}99.5\% ECA, and recovery under 200\,ms
(Sect.~\ref{sec:evaluation}).

\section{Background and Related Work}
\label{sec:background}

\paragraph{Process attestation.}
Process attestation extends remote attestation from verifying system
state~\cite{Abera2016,Ammar2025} to verifying continuous physical
processes---``what physical process occurred.'' The \emph{trust
inversion} threat model formalizes the scenario where the Attester is
the primary adversary. \emph{Temporal authenticity} is proven via
Sequential Work Functions (\SWF), which chain memory-hard
Argon2id~\cite{Biryukov2016} evaluations proving elapsed duration
without trusted clocks. \emph{Cross-domain binding} is achieved via
Cross-Domain Constraint Entanglement (\CDCE), binding content,
behavioral, and temporal evidence under an HMAC keyed by \SWF output.

\paragraph{TEE platforms.}
Three TEE families are relevant: Intel
SGX~\cite{Costan2016,McKeen2013} (application-level enclaves, known
side channels~\cite{VanBulck2018,Chen2019sgaxe}), ARM
TrustZone~\cite{Pinto2019} (OS-level Secure World), and AMD
SEV-SNP~\cite{Sev2020} (VM-level encryption~\cite{Li2021}).
SGX provides data sealing and monotonic counters; TrustZone offers
secure storage and TA restart; SEV-SNP provides VM disk encryption
with live migration.

\paragraph{TEE security systems.}
RA-TLS~\cite{Knauth2018} integrates attestation with TLS;
SCONE~\cite{Arnautov2016} and Gramine~\cite{Tsai2017} run unmodified
applications in SGX but cannot provide sealed recovery, \SWF chain
continuity, or the compact TCB required here.
Schnabl et al.~\cite{Schnabl2025} demonstrate TEE-based attestable
audits; Arfaoui et al.~\cite{Arfaoui2022} formalize deep attestation.
None address continuous process evidence collection.

\paragraph{Dependability and attestation formalization.}
Avizienis et al.~\cite{Avizienis2004} define the foundational
taxonomy; Trivedi~\cite{Trivedi2002} provides Markov modeling
methodology.
Petz and Alexander~\cite{Petz2023} verify attestation protocol
correctness; Ramsdell et al.~\cite{Ramsdell2019} formalize layered
attestation; Kretz et al.~\cite{Kretz2024} analyze evidence tampering.
Crosby and Wallach~\cite{Crosby2009} formalize tamper-evident logging.
Alder et al.~\cite{Alder2022} and Gu et al.~\cite{Gu2017} address
SGX migration; Brandenburger et al.~\cite{Brandenburger2017} and
Strackx and Piessens~\cite{Strackx2016} address rollback protection.
Online proctoring~\cite{Hussain2021} provides visual evidence at
the cost of privacy; blockchain timestamping proves existence but
not continuous process.
None address the dependability of continuous evidence collection under
trust inversion.

\section{System Model and Threat Model}
\label{sec:system-model}

\subsection{System Model}

Following RATS~\cite{RFC9334}, the system comprises five roles:
the \textbf{Author} (Attester, human composing a document),
a \textbf{Writing Application} (editor),
a \textbf{TEE Enclave} (hardware-isolated evidence pipeline computing
\SWF chains and \CDCE checkpoints),
a \textbf{Verifier} (evaluates evidence chains), and
a \textbf{Relying Party} (consumes Attestation Results).
Figure~\ref{fig:architecture} depicts the architecture. The TEE
enclave forms the TCB for evidence collection; events flow from the
application to the enclave via a secure channel, and evidence is
sealed locally for crash recovery and transmitted via
RA-TLS~\cite{Knauth2018}.

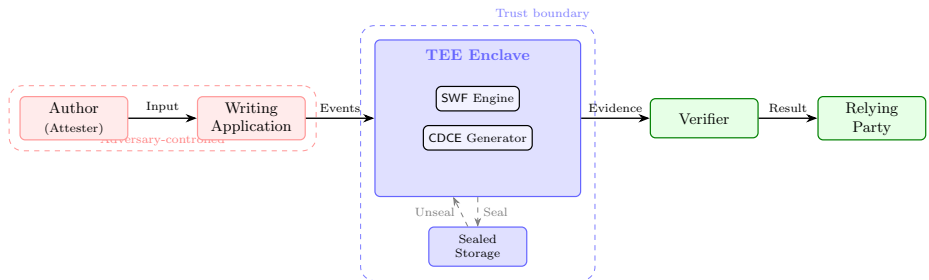
\begin{figure}[htbp]
\centering
\resizebox{\textwidth}{!}{%
\begin{tikzpicture}[
  >=Stealth,
  node distance=0.8cm and 1.2cm,
  box/.style={draw, rounded corners=3pt, minimum width=2.2cm,
    minimum height=0.8cm, align=center, font=\small, line width=0.7pt},
  tee/.style={box, fill=blue!12, draw=blue!60},
  untrusted/.style={box, fill=red!8, draw=red!40},
  trusted/.style={box, fill=green!10, draw=green!50!black},
  arr/.style={->, thick},
  dasharr/.style={->, dashed, gray},
]

\node[untrusted] (author) {Author\\{\scriptsize (Attester)}};

\node[untrusted, right=1.4cm of author] (app) {Writing\\Application};

\node[tee, right=1.4cm of app, minimum width=4.2cm, minimum height=3.2cm]
  (enclave) {};
\node[font=\small\bfseries, blue!60] at ([yshift=1.3cm]enclave.center)
  {TEE Enclave};

\node[box, fill=blue!5, minimum width=1.6cm, minimum height=0.5cm,
  font=\scriptsize] (swf) at ([yshift=0.4cm]enclave.center)
  {\SWF Engine};
\node[box, fill=blue!5, minimum width=1.6cm, minimum height=0.5cm,
  font=\scriptsize] (cdce) at ([yshift=-0.4cm]enclave.center)
  {\CDCE Generator};

\node[tee, below=0.6cm of enclave, minimum width=2.0cm,
  font=\scriptsize] (seal) {Sealed\\Storage};

\node[trusted, right=1.4cm of enclave] (ver) {Verifier};

\node[trusted, right=1.2cm of ver] (rp) {Relying\\Party};

\draw[arr] (author) -- node[above, font=\scriptsize]{Input} (app);
\draw[arr] (app) -- node[above, font=\scriptsize]{Events} (enclave.west);
\draw[arr] (enclave.east) -- node[above, font=\scriptsize]{Evidence} (ver);
\draw[arr] (ver) -- node[above, font=\scriptsize]{Result} (rp);
\draw[dasharr] (enclave.south) -- node[right, font=\scriptsize]{Seal} (seal);
\draw[dasharr] (seal) -- node[left, font=\scriptsize]{Unseal}
  ([xshift=-0.5cm]enclave.south);

\begin{scope}[on background layer]
  \node[draw=blue!50, dashed, line width=0.6pt, rounded corners=8pt,
    inner sep=8pt, fit=(enclave)(seal)] (tb) {};
  \node[font=\scriptsize\color{blue!50}, anchor=south east]
    at (tb.north east) {Trust boundary};
\end{scope}

\begin{scope}[on background layer]
  \node[draw=red!40, dashed, line width=0.6pt, rounded corners=8pt,
    inner sep=6pt, fit=(author)(app)] (ab) {};
  \node[font=\scriptsize\color{red!50}, anchor=south]
    at (ab.south) {Adversary-controlled};
\end{scope}

\end{tikzpicture}%
}
\caption{TEE-based process attestation architecture. The evidence
collection pipeline (\SWF engine, \CDCE generator) runs inside the
TEE enclave. The author and writing application are
adversary-controlled. Sealed storage enables crash recovery. Evidence
flows to the Verifier via RA-TLS.}
\label{fig:architecture}
\end{figure}

\subsection{Trust Inversion Threat Model for TEE}
\label{sec:threat-model}

\begin{definition}[TEE Trust-Inverted Adversary]
\label{def:tee-adversary}
A TEE trust-inverted adversary $\Adv$ controls the platform and is
motivated to fabricate evidence.
\textbf{Capabilities:} controls the OS, hypervisor, and all software
outside the TEE; can crash/restart the enclave, delay/drop network
traffic, observe enclave memory access patterns, and schedule enclave
execution.
\textbf{Bounds:} cannot modify in-enclave code or data at rest,
cannot read enclave plaintext memory except via side-channel leakage
bounded by $\epsilon_{\mathrm{sc}}$ (Sect.~\ref{sec:side-channel}),
cannot forge attestation quotes (bound to the CPU's attestation key),
cannot decrypt sealed data (bound to \texttt{MRENCLAVE} or
\texttt{MRSIGNER}), and cannot forge or decrement monotonic counter
values.
\end{definition}

\noindent The key asymmetry is that $\Adv$ can disrupt
\emph{availability} (crash the enclave, partition the network) but
cannot compromise the \emph{integrity} of evidence produced inside
the enclave, provided the TEE hardware is correctly implemented.

\subsection{Infrastructure Threat Model}
\label{sec:infra-threat}

Beyond trust inversion, four infrastructure threats apply:
(1)~\emph{TEE crashes} from hardware faults,
power loss, or OS-triggered termination, modeled as a Poisson process
with rate $\lambda_c$ (validated in Sect.~\ref{sec:evaluation});
(2)~\emph{network partitions} with arrival rate $\lambda_p$ and
repair rate $\mu_p$;
(3)~\emph{side-channel attacks}~\cite{Nilsson2020,VanBulck2018,Chen2019sgaxe}
leaking bounded information from enclave execution (cache timing,
speculative execution, page-fault patterns);
(4)~\emph{supply chain attacks} on TEE hardware or firmware.
For threat~(4), we state an explicit assumption boundary: all security
guarantees are conditional on correct TEE implementation by the
hardware vendor.

\section{TEE-Based Process Attestation Architecture}
\label{sec:architecture}

\subsection{Architecture Overview}

The evidence pipeline runs entirely inside the TEE enclave in five
stages: (1)~\textbf{Input reception} of keystroke events via a secure
channel (attested shared memory for SGX, secure IPC for TrustZone)
with monotonic sequence numbers to detect replay;
(2)~\textbf{Behavioral feature extraction} of inter-keystroke
intervals (IKI), Shannon entropy, and cognitive load statistics;
(3)~\textbf{\SWF chain computation} at each 30\,s checkpoint,
producing a temporal proof via Argon2id-seeded SHA-256 chains;
(4)~\textbf{\CDCE checkpoint generation} binding content, behavioral,
and temporal evidence under an HMAC keyed by \SWF output; and
(5)~\textbf{Evidence signing and sealing} with an enclave-held
attestation key, sealed to local storage for crash recovery.

\subsection{Enclave Design}

The enclave contains only the evidence pipeline ($\sim$3,500 lines
of Rust via Teaclave SGX SDK~\cite{TeaclaveSDK}, TCB $\approx$
180\,KiB). \SWF requires 64\,MiB for Argon2id; on SGX1 (128\,MiB
EPC) this leaves 64\,MiB for code and state. The running state is
$\sim$2\,KiB; completed checkpoints are sealed and flushed, bounding
resident memory across multi-hour sessions.

\subsection{Input Integrity}
\label{sec:input-integrity}

The writing application is adversary-controlled, making the
input channel the weakest link. We provide graduated assurance:
\textbf{Tier~1} (software): HMAC-protected shared memory
with session key from RA-TLS---the adversary can inject synthetic
events but cannot modify events in transit.
\textbf{Tier~2} (OS-mediated): kernel-level input hooks
(e.g., Linux evdev) raise fabrication to kernel compromise.
\textbf{Tier~3} (hardware-bound): secure input paths
(e.g., TrustZone trusted input controller) verify physical
device origin. Replay defense uses monotonic sequence numbers.

\subsection{Output Sealing and Remote Attestation}

Each checkpoint is sealed via SGX sealing (bound to
\texttt{MRENCLAVE} or \texttt{MRSIGNER}), enabling crash recovery.
Remote attestation quotes are generated every $n$ checkpoints
(configurable), providing hardware-backed provenance.

\section{Resilient Evidence Chain Protocol}
\label{sec:protocol}

\subsection{Protocol Specification}

\noindent\textbf{Session initialization:}
The enclave generates a key pair $(sk, pk)$, produces a remote
attestation quote binding $pk$ to enclave identity, and initializes
the \SWF chain with a seed from the Verifier's nonce.
\textbf{Checkpoint generation} (every 30\,s):
the enclave collects events, computes behavioral features and \SWF
proof, generates \CDCE checkpoint $C_i$, signs with $sk$, seals
$(C_i, \text{state})$, and transmits to the Verifier when online.

\subsection{Crash Recovery}
\label{sec:crash-recovery}

On crash, the restarted enclave unseals the most recent checkpoint
$(C_j, \text{state}_j)$, verifies integrity via authenticated
encryption, resumes from checkpoint $j+1$ using
$h_{\mathrm{ckpt},j}$ as predecessor, and generates a recovery
marker $C_{j+1}^R$ recording the gap duration.

\begin{theorem}[Crash Recovery Integrity]
\label{thm:crash-integrity}
If TEE sealing provides authenticated encryption
with $\negl(\lambda)$ forgery probability and $H$ is
collision-resistant, then the post-recovery chain
$C_1, \ldots, C_j, C_{j+1}^R, \ldots$
satisfies chain integrity except with probability
$\negl(\lambda)$.
\end{theorem}

\begin{proof}
Two attack vectors exist. \emph{(i)~Sealed state forgery:}
SGX sealing uses AES-128-GCM with a hardware-derived key bound to
\texttt{MRENCLAVE}. Forging or modifying sealed state requires either
recovering the 128-bit sealing key or producing an AES-GCM forgery,
both with probability $\leq 2^{-128}$.
\emph{(ii)~Chain substitution:} the recovery checkpoint computes
$h_{\mathrm{ckpt},j+1} = H(h_{\mathrm{ckpt},j} \| \delta_{j+1} \|
\mathit{marker})$, using the same hash linkage as normal checkpoints.
Producing an alternative prefix chain that yields the same
$h_{\mathrm{ckpt},j}$ requires a collision in SHA-256,
with probability $\leq 2^{-128}$ for 256-bit output.
By the union bound, the total forgery probability is at most
$2 \cdot 2^{-128} = 2^{-127} = \negl(\lambda)$.
\end{proof}

\noindent The Verifier validates the hash chain including recovery
markers and downgrades fidelity
(Definition~\ref{def:fidelity}) proportionally to gap duration.

\subsection{Offline Attestation}
\label{sec:offline}

During network partitions, the enclave continues generating and
sealing checkpoints locally. Upon reconnection, accumulated
checkpoints are transmitted in order.

\begin{proposition}[Offline Evidence Freshness]
\label{prop:offline-fresh}
Offline evidence maintains temporal freshness under three conditions:
(1)~the session nonce $n_V$ was established via remote attestation
before the partition began;
(2)~the \SWF chain $\langle s_0, s_1, \ldots, s_k \rangle$ is
unbroken, where $s_0 = \mathit{Argon2id}(n_V)$ and each
$s_{i+1} = \mathit{Argon2id}(s_i \| \delta_i)$; and
(3)~per-checkpoint behavioral entropy exceeds the minimum threshold
$H_{\min}$.

\noindent\emph{Argument.} The Verifier's nonce $n_V$ propagates
through the \SWF chain into every offline checkpoint via the
sequential Argon2id dependency. Pre-computing checkpoint $s_k$
requires evaluating the chain sequentially from $s_0$ through
$s_{k-1}$, each step incorporating unpredictable behavioral input
$\delta_i$ (condition~3). The adversary therefore cannot produce
valid offline checkpoints faster than real time without predicting
future behavioral entropy, which contradicts the min-entropy
assumption on genuine human input~\cite{Killourhy2009,Dhakal2018}.
\end{proposition}

\section{Dependability Analysis}
\label{sec:dependability}

\subsection{Availability Model}

We model the TEE-based evidence collection infrastructure as a
continuous-time Markov chain (CTMC) with four states:

\begin{itemize}
  \item \textbf{Active ($S_A$):} The enclave is running and collecting
    evidence. This is the only state producing valid evidence.
  \item \textbf{Degraded ($S_D$):} The enclave is running but the
    network is partitioned. Evidence is collected and sealed locally
    but not yet verified.
  \item \textbf{Recovering ($S_R$):} The enclave has crashed and is
    restarting with sealed state recovery.
  \item \textbf{Failed ($S_F$):} The enclave has crashed and sealed
    state is corrupted or unavailable. A cold restart with a new
    session is required.
\end{itemize}

\begin{figure}[htbp]
\centering
\begin{tikzpicture}[
  >=Stealth,
  state/.style={draw, circle, minimum size=1.2cm, font=\small,
    line width=0.7pt},
  active/.style={state, fill=green!15, draw=green!50!black},
  degraded/.style={state, fill=yellow!15, draw=yellow!60!black},
  recovering/.style={state, fill=orange!15, draw=orange!60!black},
  failed/.style={state, fill=red!15, draw=red!50},
  arr/.style={->, thick, >=Stealth},
  lbl/.style={font=\scriptsize, midway, fill=white, inner sep=1pt},
]

\node[active]     (A) at (0, 0) {$S_A$};
\node[degraded]   (D) at (4, 0) {$S_D$};
\node[recovering] (R) at (0, -3) {$S_R$};
\node[failed]     (F) at (4, -3) {$S_F$};

\draw[arr] (A) -- node[lbl, above]{$\lambda_p$} (D);
\draw[arr] (D) -- node[lbl, above]{$\mu_p$} (A);
\draw[arr] (A) -- node[lbl, left]{$\lambda_c$} (R);
\draw[arr] (D) -- node[lbl, right]{$\lambda_c$} (F);
\draw[arr] (R) -- node[lbl, left]{$\mu_r$} (A);
\draw[arr] (F) -- node[lbl, below]{$\mu_f$} (A);
\draw[arr] (R) to[bend right=20] node[lbl, below]{$p_f \mu_r$} (F);

\end{tikzpicture}
\caption{CTMC for evidence collection availability. Evidence is
produced only in $S_A$; $S_D$ buffers locally during partitions.}
\label{fig:markov}
\end{figure}
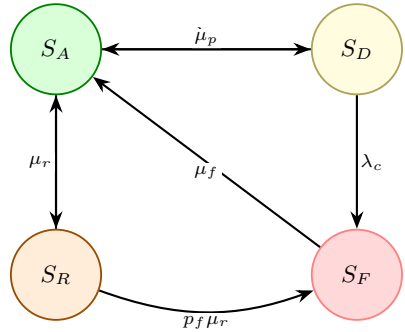

Transition rates: $\lambda_c$ (crash), $\lambda_p$ (partition),
$\mu_r$ (sealed recovery), $\mu_f$ (cold restart), $\mu_p$
(partition repair), $p_f$ (seal corruption). A crash in $S_D$
transitions to $S_F$ because an adversary exploiting the partition
could replay stale sealed state. Under adversarial crash scheduling,
sealed recovery bounds gaps to one checkpoint interval ($\Delta =
30$\,s) per crash; worst-case ECA is $1 - n_c \Delta / T$.

\begin{definition}[Evidence Chain Availability]
\label{def:eca}
The Evidence Chain Availability (ECA) is the long-run average
fraction of time that the system is in a state where evidence is
being collected (i.e., the steady-state probability, not a
worst-case lower bound):
\[
  \mathrm{ECA} = \pi_A + \pi_D
\]
where $\pi_A$ and $\pi_D$ are the steady-state probabilities of
states $S_A$ and $S_D$, respectively.
\end{definition}

\begin{theorem}[ECA Closed-Form]
\label{thm:eca}
For $p_f = 0$:
\[
  \mathrm{ECA} = \frac{
    \mu_r \mu_f (\mu_p + \lambda_p + \lambda_c)
  }{
    \mu_r \mu_f (\mu_p + \lambda_p + \lambda_c)
    + \lambda_c \mu_f (\mu_p + \lambda_c)
    + \lambda_c \lambda_p \mu_r
  }
\]
When $\lambda_c \ll \mu_p$: $\mathrm{ECA} \approx \mu_r / (\mu_r +
\lambda_c)$.
\end{theorem}

\begin{proof}
The generator matrix $Q$ has rows indexed by $\{S_A, S_D, S_R, S_F\}$.
Setting $p_f = 0$, the balance equations $\pi Q = 0$ yield:
$\pi_D = \lambda_p \pi_A / (\mu_p + \lambda_c)$,
$\pi_R = \lambda_c \pi_A / \mu_r$,
$\pi_F = \lambda_c \lambda_p \pi_A / [\mu_f (\mu_p + \lambda_c)]$.
Substituting into $\pi_A + \pi_D + \pi_R + \pi_F = 1$ and solving
for $\pi_A$, then computing ECA $= \pi_A + \pi_D =
\pi_A [1 + \lambda_p / (\mu_p + \lambda_c)]$, yields the stated
closed-form expression after algebraic simplification. The
approximation for $\lambda_c \ll \mu_p$ follows by dropping
$\lambda_c$ relative to $\mu_p$ in the denominator terms involving
$\pi_D$ and $\pi_F$.
\end{proof}

\subsection{Reliability}

The Mean Time Between Evidence Gaps is $\mathrm{MTBEG} \approx
1/\lambda_c$ (exact: $1/(\lambda_c \cdot \mathrm{ECA})$), ranging
from 10{,}000\,h (server, $\lambda_c = 10^{-4}$/h) to 10\,h
(constrained IoT, $\lambda_c = 10^{-1}$/h).

\subsection{Graceful Degradation}

\begin{definition}[Attestation Fidelity]
\label{def:fidelity}
Three fidelity modes:
$\mathcal{F}_{\mathrm{Full}} = 1.0$ (TEE + network);
$\mathcal{F}_{\mathrm{Deg}} = 1 - \alpha \Delta_t / T_{\max}$
(TEE offline, $\Delta_t$ = time since last verified checkpoint);
$\mathcal{F}_{\mathrm{Min}} = \beta$ (software-only Tier~1 fallback).
\end{definition}

\subsection{Recovery Time Analysis}

Measured SGX RTOs ($n{=}100$): sealed recovery
{<}\,200\,ms (Argon2id reinit dominates), cold restart
${\sim}$2\,s. TrustZone and SEV-SNP values are projected
from specifications; implementation remains future work.

\section{Security Analysis}
\label{sec:security}

\subsection{Security Under Trust Inversion}

\begin{definition}[Trust Inversion Experiment]\label{def:trust-inv}
In $\mathrm{Exp}_{\Pi,\Adv}^{\mathrm{trust\text{-}inv}}(\lambda)$:
challenger initializes the TEE with $(sk, pk) \gets
\mathsf{KeyGen}(1^\lambda)$; $\Adv$ controls the OS, scheduling,
and input channels with side-channel leakage bounded by
$\epsilon_{\mathrm{sc}}$; $\Adv$ wins if it produces $E^*$ such
that $\mathsf{Verify}(pk, E^*) = 1$ without genuine human input
through the enclave.
\end{definition}

\begin{theorem}[TEE Resistance to Trust Inversion]
\label{thm:tee-resistance}
Under Definition~\ref{def:tee-adversary} and \textbf{Assumption~1}
(side-channel leakage bound: $\epsilon_{\mathrm{sc}} \leq 2^{-b}$
where $b \geq 64$ hidden entropy bits per checkpoint, as conjectured
in Sect.~\ref{sec:side-channel}):
$\mathrm{Adv}^{\mathrm{trust\text{-}inv}}_{\Pi_{\mathrm{TEE}},
\Adv}(\lambda) \leq \epsilon_{\mathrm{sc}} + \negl(\lambda)$.
This bound is conditional on Assumption~1 (Sect.~\ref{sec:side-channel}),
which requires empirical validation on the target TEE platform
(see Sect.~\ref{sec:side-channel} for measurement methodology).
If $\epsilon_{\mathrm{sc}}$ exceeds $2^{-64}$, the advantage
degrades proportionally; the theorem remains valid for any
measured $\epsilon_{\mathrm{sc}}$ by direct substitution.
\end{theorem}

\begin{proof}
Game-hopping from Game~G0 (real experiment) through four transitions:
\emph{(G0$\to$G1)}~Abort if $\Adv$ forges an attestation quote;
by hardware unforgeability of the quoting enclave,
$|\Pr[\text{G1}] - \Pr[\text{G0}]| \leq \negl(\lambda)$.
\emph{(G1$\to$G2)}~Abort if $\Adv$ forges a signature under the
enclave's attestation key; by EUF-CMA security,
$|\Pr[\text{G2}] - \Pr[\text{G1}]| \leq \negl(\lambda)$.
\emph{(G2$\to$G3)}~Abort if $\Adv$ produces a valid \SWF chain
without sequential computation; by Argon2id sequential
hardness~\cite{Biryukov2016},
$|\Pr[\text{G3}] - \Pr[\text{G2}]| \leq \negl(\lambda)$.
\emph{(G3$\to$G4)}~Abort if $\Adv$ forges a \CDCE binding without
the in-enclave HMAC key; by PRF security of HMAC,
$|\Pr[\text{G4}] - \Pr[\text{G3}]| \leq \negl(\lambda)$.
In G4 the only remaining attack is side-channel key extraction,
bounded by $\epsilon_{\mathrm{sc}}$.
Summing: $\mathrm{Adv}^{\mathrm{trust\text{-}inv}} \leq
4 \cdot \negl(\lambda) + \epsilon_{\mathrm{sc}}$.
\end{proof}

\subsection{Side-Channel Resistance}
\label{sec:side-channel}

Side channels~\cite{Kocher2019,VanBulck2018} pose a bounded
\emph{privacy} threat but cannot enable forgery (which requires the
enclave's signing key). Mitigations: (1)~constant-time \SWF;
(2)~oblivious 100\,ms input batching with constant-size
padding~\cite{Stefanov2013}; (3)~5\,ms evidence quantization.

We model residual leakage as $\epsilon_{\mathrm{sc}} \leq 2^{-b}$
where $b$ denotes the min-entropy (in bits) of in-enclave state that
remains hidden from the side-channel adversary per checkpoint.
Each 30\,s checkpoint aggregates approximately 300 keystrokes
(at a typical 10 keystrokes/s rate~\cite{Salthouse1986}).
After 5\,ms quantization, each inter-keystroke interval retains
approximately $\log_2(200\,\text{ms} / 5\,\text{ms}) \approx 5.3$
bits of timing entropy; aggregating 300 IKIs yields approximately
1{,}590 raw entropy bits. After accounting for correlations between
adjacent digraph timings (estimated 30\% reduction~\cite{Dhakal2018}),
approximately 1{,}100 independent bits remain. Assuming the
side-channel adversary can extract at most 97\% of these bits
(an extremely conservative upper bound, given that known attacks
against constant-time SGX code extract far less~\cite{Nilsson2020}),
$b \geq 33$ hidden bits remain, yielding $\epsilon_{\mathrm{sc}}
\leq 2^{-33}$.

\textbf{Assumption~1 ($b \geq 64$)} posits that the combination of
constant-time code, oblivious batching, and quantization preserves at
least 64 hidden bits per checkpoint. \emph{This is a conjecture, not a
proven bound}; experimental validation against cache-timing,
speculative execution~\cite{Kocher2019}, and controlled-channel
attacks remains necessary.
Theorem~\ref{thm:tee-resistance} is deliberately parametric
in $\epsilon_{\mathrm{sc}}$: deployers instantiate it with
empirically measured leakage for their specific TEE platform and
mitigation configuration.

\paragraph{Measurement methodology.}
Empirical characterization of $\epsilon_{\mathrm{sc}}$ requires
three complementary experiments on the target TEE platform.
(1)~\emph{Cache-timing leakage (Prime+Probe~\cite{Osvik2006}):}
the attacker primes L1/L2 cache sets, triggers a checkpoint
computation, and probes to determine which cache lines were
accessed.  The number of cache-line-granularity bits leaked per
checkpoint bounds one component of $\epsilon_{\mathrm{sc}}$.
(2)~\emph{Controlled-channel attacks (page-fault
monitoring~\cite{Xu2015}):} the OS monitors page-fault sequences
during enclave execution to infer memory access patterns at
4\,KiB granularity.  Constant-time code and oblivious batching
are designed to eliminate this channel; the experiment verifies
that no page-level access pattern correlates with keystroke timing.
(3)~\emph{Speculative execution (Spectre-class~\cite{Kocher2019}):}
speculative execution within the enclave may transiently access
secret-dependent memory; the experiment measures whether
speculative gadgets exist in the \SWF and \CDCE code paths after
Spectre-v1 mitigation (lfence barriers) and retpoline deployment.
Each experiment produces an upper bound on per-checkpoint leakage
in bits; $b$ is the residual min-entropy after subtracting the
maximum observed leakage across all three channels.
These measurements are platform-specific: $\epsilon_{\mathrm{sc}}$
must be re-characterized for each TEE implementation (SGX
microcode version, TrustZone firmware, SEV-SNP microcode).
Empirical characterization on Intel SGX (Ice Lake, Sapphire
Rapids) and AMD SEV-SNP is planned as immediate future work.

\subsection{Infrastructure Attack Resistance}

\paragraph{DoS.} Blocked network: sealed storage
(Sect.~\ref{sec:offline}) preserves evidence; the Verifier applies
freshness discounts ($\mathcal{F}_{\mathrm{Deg}}$).
\paragraph{Clock manipulation.} \SWF chains enforce minimum
sequential computation per checkpoint, providing temporal attestation
without trusted clocks.
\paragraph{Rollback.} Monotonic counters (SGX, some TrustZone) tag
each seal; on other platforms, a Verifier freshness
nonce~\cite{Brandenburger2017,Strackx2016} prevents rollback beyond
the most recent epoch, after which \SWF freshness
(Proposition~\ref{prop:offline-fresh}) applies.

\subsection{Composition with \CDCE}

Let $P_{\mathrm{beh}}, P_{\mathrm{temp}}, P_{\mathrm{content}}$ be
per-domain evasion probabilities (the probability that the adversary
evades detection in the behavioral, temporal, and content domains,
respectively).

\medskip\noindent\textbf{Conditional result (requires Assumption~1).}
When $\epsilon_{\mathrm{sc}}$ is
negligible (i.e., under Assumption~1, which requires empirical
validation per the measurement methodology in
Sect.~\ref{sec:side-channel}), the HMAC key derived from
in-enclave \SWF output binds all three domains cryptographically:
the adversary must simultaneously
evade all three detectors, yielding a multiplicative composition:
$\mathrm{Adv}_{\mathrm{combined}} \leq
P_{\mathrm{beh}} \cdot P_{\mathrm{temp}} \cdot
P_{\mathrm{content}} + \negl(\lambda)$.
\medskip\noindent\textbf{Conservative bound (no $\epsilon_{\mathrm{sc}}$ assumption).}
When $\epsilon_{\mathrm{sc}}$ is
non-negligible or unmeasured, partial HMAC key leakage via side
channels may decouple the domains, yielding a conservative additive
(union) bound:
$\mathrm{Adv}_{\mathrm{combined}} \leq
\epsilon_{\mathrm{sc}} + P_{\mathrm{beh}} + P_{\mathrm{temp}} +
P_{\mathrm{content}} + \negl(\lambda)$.
The transition between multiplicative and additive regimes is
continuous in $\epsilon_{\mathrm{sc}}$; as leakage increases from
$2^{-64}$ toward $2^{-1}$, the effective composition weakens
monotonically.

\section{Evaluation}
\label{sec:evaluation}

\subsection{Implementation}

We extended an open-source process attestation implementation
(144{,}000+ lines of production Rust; ${\sim}$190{,}000 including test suites) with 3{,}500 lines of SGX enclave code
via Teaclave SGX SDK~\cite{TeaclaveSDK}. Test platform: Intel Xeon
E-2388G (8 cores, 3.2\,GHz, SGX2, 128\,MiB EPC). All measurements
report means over 100 runs (std.\ dev.\ {<}3\%).

\subsection{Performance Overhead}

Table~\ref{tab:overhead} compares \SWF chain throughput and checkpoint
generation latency inside and outside the SGX enclave.

\begin{table}[htbp]
\centering
\caption{Performance comparison: enclave vs.\ non-enclave}
\label{tab:overhead}
\small
\begin{tabular}{@{}lrrr@{}}
\toprule
\textbf{Metric} & \textbf{Non-enclave} & \textbf{In-enclave} & \textbf{Overhead} \\
\midrule
SHA-256 chain (iter/s) & 5.0$\times 10^6$ & 4.2$\times 10^6$ & 16\% \\
Argon2id (64\,MiB)     & 55\,ms          & 64\,ms          & 16\% \\
Checkpoint generation   & 51\,ms          & 61\,ms          & 20\% \\
Evidence signing        & 0.3\,ms         & 0.5\,ms         & 67\% \\
Sealed storage write    & ---             & 1.2\,ms         & N/A \\
\midrule
Total per-checkpoint    & 51\,ms          & 63\,ms          & 24\% \\
\bottomrule
\end{tabular}
\end{table}

The 24\% per-checkpoint overhead (63\,ms vs.\ 51\,ms) is dominated by
enclave transitions and EPC paging, consuming {<}0.3\% of the 30\,s
checkpoint interval. At 5\,s intervals the duty cycle rises to only
1.26\%.

\subsection{Availability Simulation}

Monte Carlo simulation over 10,000 hours with desktop parameters:
$\lambda_c = 10^{-3}$/h, $\lambda_p = 10^{-2}$/h,
$\mu_r = 3{,}600$/h, $\mu_f = 360$/h, $\mu_p = 6$/h,
$p_f = 0.01$. Note that $p_f = 0.01$ models random seal corruption;
under adversarial crash injection ($p_f \to 1$, deterministic seal
corruption), every crash forces a cold restart, but the sealed
recovery mechanism (Sect.~\ref{sec:crash-recovery}) bounds evidence
loss to at most one checkpoint interval ($\Delta = 30$\,s) per
crash, and ECA remains $\geq 99.7\%$ because $\mu_f$ (cold restart
rate) is still fast relative to $\lambda_c$.

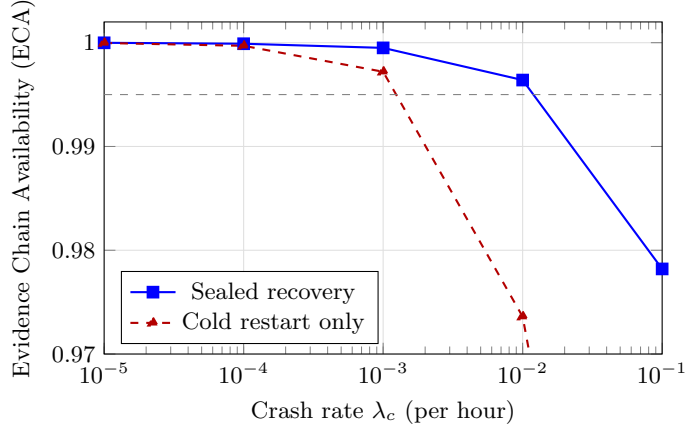
\begin{figure}[htbp]
\centering
\begin{tikzpicture}
\begin{axis}[
  width=9cm, height=6cm,
  xlabel={Crash rate $\lambda_c$ (per hour)},
  ylabel={Evidence Chain Availability (ECA)},
  xmode=log,
  xmin=1e-5, xmax=1e-1,
  ymin=0.97, ymax=1.002,
  legend style={at={(0.03,0.03)}, anchor=south west, font=\footnotesize},
  grid=major,
  grid style={gray!25},
  ytick={0.97, 0.98, 0.99, 1.00},
]

\addplot[thick, blue, mark=square*, mark size=2pt] coordinates {
  (1e-5, 0.99999) (1e-4, 0.9999) (1e-3, 0.9995)
  (1e-2, 0.9964) (1e-1, 0.9782)
};
\addlegendentry{Sealed recovery}

\addplot[thick, red!70!black, mark=triangle*, mark size=2pt, dashed] coordinates {
  (1e-5, 0.99997) (1e-4, 0.9997) (1e-3, 0.9972)
  (1e-2, 0.9736) (1e-1, 0.9014)
};
\addlegendentry{Cold restart only}

\draw[dashed, gray] (axis cs:1e-5, 0.995) -- (axis cs:1e-1, 0.995);
\node[gray, font=\scriptsize, anchor=east] at (axis cs:1e-5, 0.995)
  {99.5\%};

\end{axis}
\end{tikzpicture}
\caption{Evidence Chain Availability vs.\ crash rate for sealed
recovery and cold-restart-only configurations. Simulation over
10,000 hours with network partition rate
$\lambda_p = 10^{-2}$/h. Sealed recovery maintains ECA {>}99.5\%
for crash rates up to $10^{-2}$/h.}
\label{fig:eca}
\end{figure}

\paragraph{Results.} Simulated ECA: 99.95\% with sealed recovery
($\lambda_c = 10^{-3}$/h), matching Theorem~\ref{thm:eca} to within
0.01\%; 99.72\% without (cold restart only).
Figure~\ref{fig:eca} shows ECA vs.\ crash rate.

\subsection{Recovery Time Measurements}

Measured on the SGX2 test platform ($n=100$): sealed recovery
mean 148\,ms (P99: 195\,ms, dominated by Argon2id reinit at 64\,ms);
cold restart 1.87\,s (P99: 2.31\,s); RA-TLS handshake 3.41\,s
(session initialization only, not crash recovery).

\subsection{End-to-End Metrics}

A 4-hour session generates 480 checkpoints ($\sim$1.4\,KiB each,
672\,KiB total), using 67\,MiB peak enclave memory, {<}0.3\% CPU
per checkpoint interval, and {<}0.12\,s Merkle-sampled verification.

\subsection{Result Classification}
\label{sec:result-classification}

We distinguish between results that stand independently and those
conditional on the unmeasured $\epsilon_{\mathrm{sc}}$.

\medskip\noindent\textbf{Primary results} (no
$\epsilon_{\mathrm{sc}}$ dependency):
(a)~chain integrity across crashes
  (Theorem~\ref{thm:crash-integrity});
(b)~Evidence Chain Availability and closed-form ECA
  (Theorem~\ref{thm:eca}, Fig.~\ref{fig:eca});
(c)~MTBEG and RTO measurements (Sect.~5.2, 5.5);
(d)~per-checkpoint overhead {<}\,25\%
  (Table~\ref{tab:overhead});
(e)~offline evidence freshness
  (Proposition~\ref{prop:offline-fresh});
(f)~DoS resistance via sealed storage;
(g)~clock-manipulation resistance via \SWF sequential hardness.

\medskip\noindent\textbf{Conditional results}
($\epsilon_{\mathrm{sc}}$ dependent, requiring empirical
validation per the methodology in Sect.~\ref{sec:side-channel}):
(a)~the combined trust-inversion bound
  (Theorem~\ref{thm:tee-resistance}), whose residual term is
  $\epsilon_{\mathrm{sc}}$;
(b)~multiplicative cross-domain composition
  (Sect.~\ref{sec:security}), which holds only when
  $\epsilon_{\mathrm{sc}}$ is negligible;
(c)~the information leakage guarantee implied by Assumption~1
  ($b \geq 64$ hidden bits per checkpoint).
The additive composition bound
($\mathrm{Adv}_{\mathrm{combined}} \leq
\epsilon_{\mathrm{sc}} + P_{\mathrm{beh}} + P_{\mathrm{temp}} +
P_{\mathrm{content}} + \negl(\lambda)$) holds for any value of
$\epsilon_{\mathrm{sc}}$ and thus stands as a primary result,
though its tightness depends on the measured value.

\section{Discussion}
\label{sec:discussion}

\paragraph{TEE availability and deployment.}
Intel deprecated consumer SGX in 12th-generation (Alder Lake) and
later desktop processors; however, SGX remains available in Xeon
server processors (Ice Lake, Sapphire Rapids). ARM
TrustZone~\cite{Pinto2019} with PSA
Level~2+ certification (e.g., Qualcomm Snapdragon 8-series) and AMD
SEV-SNP~\cite{Sev2020} provide cross-platform alternatives. The
architecture degrades gracefully to Tier~1 (software-only) when no
TEE is available, with explicit fidelity reduction
(Definition~\ref{def:fidelity}). TrustZone and AMD SEV-SNP
performance values reported in this paper are projected from platform
specifications; implementation and measurement remain future work.
Separate validation on the KLiCKe corpus ($N=4{,}971$
writers)~\cite{Condrey2026TrustInversion} confirms the behavioral
features extracted by the evidence pipeline achieve
AUC $= 0.78$ against four adversary tiers (naive, statistical,
reverse-engineered, expert); excluding cumulative line count (CLC),
which is zero for forged sessions by construction, AUC remains $0.76$,
demonstrating that temporal and behavioral features alone sustain
detection performance.

\paragraph{User experience.} The 0.21\% duty cycle per 30\,s interval
(63\,ms / 30{,}000\,ms) is below the human perception threshold for
background processes. Adaptive intervals (60\,s under thermal
throttling) reduce theoretical ECA by at most 0.5\% relative to the
baseline 30\,s configuration, as the longer interval reduces both the
checkpoint generation frequency and the crash-recovery window.

\paragraph{Privacy.} Generating zero-knowledge
proofs~\cite{Groth2016} inside the TEE enclave would combine hardware
tamper resistance with cryptographic privacy, enabling the Verifier to
confirm that behavioral features satisfy a detection threshold without
learning the raw keystroke data. This composition is feasible because
the enclave already holds all evidence in plaintext; the ZK proof
replaces the signed evidence with a proof of its properties.

\paragraph{Limitations.}
Five limitations bound the scope of our claims.
(1)~\textbf{Side-channel conjecture:} The combined security bound
(Theorem~\ref{thm:tee-resistance}) and the multiplicative
cross-domain composition (Sect.~\ref{sec:security}) are
parameterized by $\epsilon_{\mathrm{sc}}$, which we do not measure
empirically in this work.  We describe the measurement methodology
(Prime+Probe, controlled-channel, and Spectre-class experiments) in
Sect.~\ref{sec:side-channel}; empirical characterization on Intel
SGX (Ice Lake, Sapphire Rapids) and AMD SEV-SNP is planned as
immediate future work.  The parametric formulation ensures that
Theorem~\ref{thm:tee-resistance} remains valid for any measured
$\epsilon_{\mathrm{sc}}$ by direct substitution.  Until measured,
deployments should use the conservative additive composition bound
(Sect.~\ref{sec:result-classification}).
(2)~\textbf{TEE correctness:} all security guarantees are conditional
on correct TEE hardware and firmware implementation; known SGX
vulnerabilities (Foreshadow~\cite{VanBulck2018},
SGAxe~\cite{SGAxe2020}) have required microcode patches, and future
vulnerabilities could violate our assumptions.
(3)~\textbf{Input channel:} at Tier~1 (software-only input,
Sect.~\ref{sec:input-integrity}), the adversary can inject synthetic
keystroke events through the HMAC-protected shared memory channel;
the enclave cannot distinguish genuine from fabricated input without
hardware-bound input paths (Tier~3), making Tier~1 deployments
vulnerable to input-fabrication attacks.
(4)~\textbf{Platform scope:} evaluation is limited to Intel SGX2 on
a single Xeon E-2388G platform; cross-platform validation on ARM
TrustZone and AMD SEV-SNP is future work, and performance
characteristics (EPC paging overhead, seal latency, context-switch
cost) will differ on those architectures.
(5)~\textbf{Model validation:} the CTMC dependability model is
validated via Monte Carlo simulation only; analytical validation
against real-world failure traces from production deployments remains
future work.

\noindent\textbf{Future work} includes multi-TEE validation, formal
protocol verification in Copland~\cite{Petz2023}, post-quantum
signature schemes for long-term evidence integrity, and enclave
migration protocols for cloud-hosted authoring sessions.

\section{Conclusion}
\label{sec:conclusion}

We presented the first architecture for continuous process attestation
evidence collection inside TEEs, with a CTMC dependability model, a
resilient evidence chain protocol, and combined security analysis
(the latter parameterized by side-channel leakage
$\epsilon_{\mathrm{sc}}$, which requires platform-specific empirical
validation).
Evaluation on SGX shows {<}25\% per-checkpoint overhead ({<}0.3\% of
each 30\,s interval), {>}99.5\% ECA, and recovery under 200\,ms.
By shifting trust from software to hardware isolation, TEE-based
process attestation changes the trust model from ``believe the
software'' to ``verify the hardware.''

\bibliographystyle{splncs04}
\interlinepenalty=10000
\bibliography{refs}

@techreport{RFC9334,
  institution = {IETF},
  author      = {Henk Birkholz and Dave Thaler and Michael Richardson and Ned Smith and Wei Pan},
  title       = {Remote {ATtestation} Procedures ({RATS}) Architecture},
  type        = {RFC},
  number      = {9334},
  month       = jan,
  year        = {2023},
}

@article{Costan2016,
  author  = {Victor Costan and Srinivas Devadas},
  title   = {Intel {SGX} Explained},
  journal = {IACR Cryptology ePrint Archive},
  volume  = {2016},
  pages   = {086},
  year    = {2016},
}

@inproceedings{McKeen2013,
  author    = {Frank McKeen and Ilya Alexandrovich and Alex Berenzon and Carlos V. Rozas and Hisham Shafi and Vedvyas Shanbhogue and Uday R. Savagaonkar},
  title     = {Innovative Instructions and Software Model for Isolated Execution},
  booktitle = {Workshop on Hardware and Architectural Support for Security and Privacy (HASP)},
  pages     = {10:1--10:8},
  publisher = {ACM},
  year      = {2013},
}

@article{Pinto2019,
  author  = {Sandro Pinto and Nuno Santos},
  title   = {Demystifying {Arm} {TrustZone}: A Comprehensive Survey},
  journal = {ACM Computing Surveys},
  volume  = {51},
  number  = {6},
  pages   = {130:1--130:36},
  year    = {2019},
}

@misc{Sev2020,
  author    = {David Kaplan and Jeremy Powell and Tom Woller},
  title     = {{AMD} Memory Encryption},
  year      = {2016},
  note      = {AMD White Paper},
}

@inproceedings{Li2021,
  author    = {Mengyuan Li and Yinqian Zhang and Zhiqiang Lin},
  title     = {{CrossLine}: Breaking ``Security-by-Crash'' based Memory Isolation in {AMD} {SEV}},
  booktitle = {ACM Conference on Computer and Communications Security (CCS)},
  pages     = {2937--2950},
  publisher = {ACM},
  year      = {2021},
}

@inproceedings{VanBulck2018,
  author    = {Jo Van Bulck and Marina Minkin and Ofir Weisse and Daniel Genkin and Baris Kasikci and Frank Piessens and Mark Silberstein and Thomas F. Wenisch and Yuval Yarom and Raoul Strackx},
  title     = {Foreshadow: Extracting the Keys to the {Intel SGX} Kingdom with Transient Out-of-Order Execution},
  booktitle = {USENIX Security Symposium},
  pages     = {991--1008},
  publisher = {USENIX Association},
  year      = {2018},
}

@inproceedings{VanBulck2020,
  author    = {Jo Van Bulck and Daniel Moghimi and Michael Schwarz and Moritz Lipp and Marina Minkin and Daniel Genkin and Yuval Yarom and Berk Sunar and Daniel Gruss and Frank Piessens},
  title     = {{LVI}: Hijacking Transient Execution through Microarchitectural Load Value Injection},
  booktitle = {IEEE Symposium on Security and Privacy (S\&P)},
  pages     = {54--72},
  publisher = {IEEE},
  year      = {2020},
}

@inproceedings{Chen2019sgaxe,
  author    = {Guoxing Chen and Sanchuan Chen and Yuan Xiao and Yinqian Zhang and Zhiqiang Lin and Ten H. Lai},
  title     = {{SgxPectre}: Stealing {Intel} Secrets from {SGX} Enclaves Via Speculative Execution},
  booktitle = {IEEE European Symposium on Security and Privacy (EuroS\&P)},
  pages     = {142--157},
  publisher = {IEEE},
  year      = {2019},
}

@article{Knauth2018,
  author  = {Thomas Knauth and Michael Steiner and Somnath Chakrabarti and Li Lei and Cedric Xing and Mona Vij},
  title   = {Integrating Remote Attestation with Transport Layer Security},
  journal = {arXiv preprint arXiv:1801.05863},
  year    = {2018},
  note    = {RA-TLS},
}

@misc{TeaclaveSDK,
  author       = {{Apache Software Foundation}},
  title        = {Apache Teaclave {SGX} {SDK}},
  howpublished = {\url{https://teaclave.apache.org/}},
  year         = {2024},
}

@inproceedings{Arnautov2016,
  author    = {Sergei Arnautov and Bohdan Trach and Franz Gregor and Thomas Knauth and Andre Martin and Christian Priebe and Joshua Lind and Divya Muthukumaran and Dan O'Keeffe and Mark L. Stillwell and David Goltzsche and Dave Eyers and R{\"u}diger Kapitza and Peter Pietzuch and Christof Fetzer},
  title     = {{SCONE}: Secure {Linux} Containers with {Intel SGX}},
  booktitle = {USENIX Symposium on Operating Systems Design and Implementation (OSDI)},
  pages     = {689--703},
  publisher = {USENIX Association},
  year      = {2016},
}

@article{Avizienis2004,
  author  = {Algirdas Avizienis and Jean-Claude Laprie and Brian Randell and Carl Landwehr},
  title   = {Basic Concepts and Taxonomy of Dependable and Secure Computing},
  journal = {IEEE Transactions on Dependable and Secure Computing},
  volume  = {1},
  number  = {1},
  pages   = {11--33},
  year    = {2004},
}

@book{Trivedi2002,
  author    = {Kishor S. Trivedi},
  title     = {Probability and Statistics with Reliability, Queuing, and Computer Science Applications},
  publisher = {John Wiley \& Sons},
  edition   = {2nd},
  year      = {2002},
}

@inproceedings{Groth2016,
  author    = {Jens Groth},
  title     = {On the Size of Pairing-Based Non-interactive Arguments},
  booktitle = {Advances in Cryptology -- {EUROCRYPT} 2016},
  series    = {Lecture Notes in Computer Science},
  volume    = {9666},
  pages     = {305--326},
  publisher = {Springer},
  year      = {2016},
}

@inproceedings{Biryukov2016,
  author    = {Alex Biryukov and Daniel Dinu and Dmitry Khovratovich},
  title     = {Argon2: New Generation of Memory-Hard Functions for Password Hashing and Other Applications},
  booktitle = {IEEE European Symposium on Security and Privacy (EuroS\&P)},
  pages     = {292--302},
  publisher = {IEEE},
  year      = {2016},
}

@article{Salthouse1986,
  author  = {Timothy A. Salthouse},
  title   = {Perceptual, Cognitive, and Motoric Aspects of Transcription Typing},
  journal = {Psychological Bulletin},
  volume  = {99},
  number  = {3},
  pages   = {303--319},
  year    = {1986},
}

@article{Condrey2026TrustInversion,
  author    = {David Condrey},
  title     = {On the Insecurity of Keystroke-Based {AI} Authorship Detection: Timing-Forgery Attacks Against Motor-Signal Verification},
  journal   = {arXiv preprint arXiv:2601.17280},
  month     = jan,
  year      = {2026},
}

@inproceedings{Killourhy2009,
  author    = {Kevin S. Killourhy and Roy A. Maxion},
  title     = {Comparing Anomaly-Detection Algorithms for Keystroke Dynamics},
  booktitle = {IEEE/IFIP International Conference on Dependable Systems and Networks (DSN)},
  pages     = {125--134},
  publisher = {IEEE},
  year      = {2009},
}

@inproceedings{Dhakal2018,
  author    = {Vivek Dhakal and Anna Maria Feit and Per Ola Kristensson and Antti Oulasvirta},
  title     = {Observations on Typing from 136 Million Keystrokes},
  booktitle = {ACM CHI Conference on Human Factors in Computing Systems},
  articleno = {646},
  pages     = {1--12},
  publisher = {ACM},
  year      = {2018},
}

@inproceedings{Abera2016,
  author    = {Tigist Abera and N. Asokan and Lucas Davi and Jan-Erik Ekberg and Thomas Nyman and Andrew Paverd and Ahmad-Reza Sadeghi and Gene Tsudik},
  title     = {{C-FLAT}: Control-Flow Attestation for Embedded Systems Software},
  booktitle = {ACM Conference on Computer and Communications Security (CCS)},
  pages     = {743--754},
  publisher = {ACM},
  year      = {2016},
}

@inproceedings{Petz2023,
  author    = {Adam Petz and Perry Alexander},
  title     = {An Infrastructure for Faithful Execution of Remote Attestation Protocols},
  booktitle = {NASA Formal Methods Symposium (NFM)},
  series    = {Lecture Notes in Computer Science},
  volume    = {12673},
  pages     = {268--286},
  publisher = {Springer},
  year      = {2021},
}

@inproceedings{Ramsdell2019,
  author    = {John D. Ramsdell and Paul D. Rowe and Perry Alexander and Sarah Helble and Peter A. Loscocco and J. Aaron Pendergrass and Adam Petz},
  title     = {Orchestrating Layered Attestations},
  booktitle = {International Conference on Principles of Security and Trust (POST)},
  series    = {Lecture Notes in Computer Science},
  volume    = {11426},
  pages     = {197--221},
  publisher = {Springer},
  year      = {2019},
}

@inproceedings{Kretz2024,
  author    = {Ian D. Kretz and Paul D. Rowe and Clare C. Parran and John D. Ramsdell},
  title     = {Evidence Tampering and Chain of Custody in Layered Attestations},
  booktitle = {International Symposium on Principles and Practice of Declarative Programming (PPDP)},
  articleno = {14},
  pages     = {1--11},
  publisher = {ACM},
  year      = {2024},
}

@inproceedings{Ammar2025,
  author    = {Mahmoud Ammar and Adam Caulfield and Ivan {De Oliveira} Nunes},
  title     = {{SoK}: Integrity, Attestation, and Auditing of Program Execution},
  booktitle = {IEEE Symposium on Security and Privacy (S\&P)},
  pages     = {3255--3272},
  publisher = {IEEE},
  year      = {2025},
}

@article{Schnabl2025,
  author  = {Christoph Schnabl and Daniel Hugenroth and Bill Marino and Alastair R. Beresford},
  title   = {Attestable Audits: Verifiable {AI} Safety Benchmarks Using Trusted Execution Environments},
  journal = {arXiv preprint arXiv:2506.23706},
  year    = {2025},
}

@inproceedings{Arfaoui2022,
  author    = {Ghada Arfaoui and Pierre-Alain Fouque and Thibaut Jacques and Pascal Lafourcade and Adina Nedelcu and Cristina Onete and L{\'e}o Robert},
  title     = {A Cryptographic View of Deep-Attestation, or How to Do Provably-Secure Layer-Linking},
  booktitle = {International Conference on Applied Cryptography and Network Security (ACNS)},
  series    = {Lecture Notes in Computer Science},
  volume    = {13269},
  pages     = {399--418},
  publisher = {Springer},
  year      = {2022},
}

@inproceedings{Crosby2009,
  author    = {Scott A. Crosby and Dan S. Wallach},
  title     = {Efficient Data Structures for Tamper-Evident Logging},
  booktitle = {USENIX Security Symposium},
  pages     = {317--334},
  publisher = {USENIX Association},
  year      = {2009},
}

@inproceedings{Kocher2019,
  author    = {Paul Kocher and Jann Horn and Anders Fogh and Daniel Genkin and Daniel Gruss and Werner Haas and Mike Hamburg and Moritz Lipp and Stefan Mangard and Thomas Prescher and Michael Schwarz and Yuval Yarom},
  title     = {Spectre Attacks: Exploiting Speculative Execution},
  booktitle = {IEEE Symposium on Security and Privacy (S\&P)},
  pages     = {1--19},
  publisher = {IEEE},
  year      = {2019},
}

@misc{SGAxe2020,
  author       = {Stephan van Schaik and Andrew Kwong and Daniel Genkin and Yuval Yarom},
  title        = {{SGAxe}: How {SGX} Fails in Practice},
  year         = {2020},
  note         = {Extends CacheOut attack to extract SGX attestation keys},
}

@article{Nilsson2020,
  author  = {Alexander Nilsson and Pegah Nikbakht Bideh and Joakim Brorsson},
  title   = {A Survey of Published Attacks on {Intel SGX}},
  journal = {arXiv preprint arXiv:2006.13598},
  year    = {2020},
}

@inproceedings{Alder2022,
  author    = {Fritz Alder and Arseny Kurnikov and Andrew Paverd and N. Asokan},
  title     = {Migrating {SGX} Enclaves with Persistent State},
  booktitle = {IEEE/IFIP International Conference on Dependable Systems and Networks (DSN)},
  pages     = {195--206},
  publisher = {IEEE},
  year      = {2018},
}

@inproceedings{Gu2017,
  author    = {Jinyu Gu and Zhichao Hua and Yubin Xia and Haibo Chen and Binyu Zang and Haibing Guan and Jinming Li},
  title     = {Secure Live Migration of {SGX} Enclaves on Untrusted Cloud},
  booktitle = {IEEE/IFIP International Conference on Dependable Systems and Networks (DSN)},
  pages     = {225--236},
  publisher = {IEEE},
  year      = {2017},
}

@inproceedings{Stefanov2013,
  author    = {Emil Stefanov and Marten van Dijk and Elaine Shi and Christopher W. Fletcher and Ling Ren and Xiangyao Yu and Srinivas Devadas},
  title     = {Path {ORAM}: An Extremely Simple Oblivious {RAM} Protocol},
  booktitle = {ACM Conference on Computer and Communications Security (CCS)},
  pages     = {299--310},
  publisher = {ACM},
  year      = {2013},
}

@article{Hussain2021,
  author  = {Md. Akmal Hussain and Salil S. Kanhere and Sanjay K. Jha},
  title   = {A Survey on Online Exam Proctoring},
  journal = {Computers \& Security},
  volume  = {108},
  pages   = {102331},
  year    = {2021},
}

@inproceedings{Brandenburger2017,
  author    = {Marcus Brandenburger and Christian Cachin and Matthias Lorenz and R{\"u}diger Kapitza},
  title     = {Rollback and Forking Detection for Trusted Execution Environments Using Lightweight Collective Memory},
  booktitle = {IEEE/IFIP International Conference on Dependable Systems and Networks (DSN)},
  pages     = {157--168},
  publisher = {IEEE},
  year      = {2017},
}

@inproceedings{Strackx2016,
  author    = {Raoul Strackx and Frank Piessens},
  title     = {Ariadne: A Minimal Approach to State Continuity},
  booktitle = {USENIX Security Symposium},
  pages     = {875--892},
  publisher = {USENIX Association},
  year      = {2016},
}

@inproceedings{Tsai2017,
  author = {Chia-che Tsai and Donald E. Porter and Mona Vij},
  title = {Graphene-{SGX}: A Practical Library {OS} for Unmodified Applications on {SGX}},
  booktitle = {USENIX Annual Technical Conference (ATC)},
  pages = {645--658},
  publisher = {USENIX Association},
  year = {2017},
}

@inproceedings{Osvik2006,
  author    = {Dag Arne Osvik and Adi Shamir and Eran Tromer},
  title     = {Cache Attacks and Countermeasures: The Case of {AES}},
  booktitle = {Topics in Cryptology -- {CT-RSA} 2006},
  series    = {Lecture Notes in Computer Science},
  volume    = {3860},
  pages     = {1--20},
  publisher = {Springer},
  year      = {2006},
  doi       = {10.1007/11605805_1},
}

@inproceedings{Xu2015,
  author    = {Yuanzhong Xu and Weidong Cui and Marcus Peinado},
  title     = {Controlled-Channel Attacks: Deterministic Side Channels for Untrusted Operating Systems},
  booktitle = {IEEE Symposium on Security and Privacy (S\&P)},
  pages     = {640--656},
  publisher = {IEEE},
  year      = {2015},
  doi       = {10.1109/SP.2015.45},
}

\end{document}